 \documentclass[a4paper]{article}
\usepackage{amsthm}

\theoremstyle{plain}
\newtheorem*{theorem*}{Theorem}
\newtheorem*{definition*}{Definition}

\theoremstyle{plain}
\newtheorem{thm}{\protect\theoremname}
  \theoremstyle{plain}
  
  \theoremstyle{plain}
  \newtheorem{lemma}[thm]{\protect\lemmaname}
  \theoremstyle{plain}
  
  \theoremstyle{plain}
  
  \theoremstyle{plain}
  \newtheorem*{lem*}{\protect\lemmaname}
  \newtheorem*{cor*}{\protect\corollaryname}
 
\usepackage{microtype}%if unwanted, comment out or use option "draft"
\usepackage{amsmath}
\usepackage{amssymb}

\usepackage{thmtools}
\usepackage{thm-restate}
\usepackage{cleveref}

\usepackage{algorithm}
\usepackage{algorithmic}

  \providecommand{\lemmaname}{Lemma}
\providecommand{\theoremname}{Theorem}

\makeatother

  \providecommand{\corollaryname}{\inputencoding{latin9}Corollary}
  \providecommand{\factname}{\inputencoding{latin9}Fact}
  \providecommand{\lemmaname}{\inputencoding{latin9}Lemma}
\providecommand{\theoremname}{\inputencoding{latin9}Theorem}

\begin{document}

% Author macros %%%%%%%%%%%%%%%%%%%%%%%%%%%%%%%%%%%%%%%%%%%%%%%%
\title{Submodular Stochastic Probing on Matroids}

\author{Marek Adamczyk
\thanks{Department of Computer, Control, and Management Engineering, Sapienza University of Rome, Italy; e-mail: \texttt{adamczyk@dis.uniroma1.it}}
\and Maxim Sviridenko
\thanks{Department of Computer Science, University of Warwick, United Kingdom, 			\texttt{M.I.Sviridenko@warwick.ac.uk}}
\and Justin Ward
\thanks{Department of Computer Science, University of Warwick, United Kingdom, 
			\texttt{J.D.Ward@warwick.ac.uk}}}

%\Copyright{Marek Adamczyk and Maxim Sviridenko and Justin Ward}%mandatory. LIPIcs license is "CC-BY";  http://creativecommons.org/licenses/by/3.0/

%\subjclass{Theory of computation Stochastic approximation}% mandatory: Please choose ACM 1998 classifications from http://www.acm.org/about/class/ccs98-html . E.g., cite as "F.1.1 Models of Computation". 
%\keywords{approximation algorithms, stochastic optimization, submodular optimization, matroids, iterative rounding}% mandatory: Please provide 1-5 keywords

\global\long\def\opstyle#1{\mathbb{#1}}

\global\long\def\ex#1{\mathbb{E}\left[#1\right]}

\global\long\def\xp#1{\mathbb{E}\left[#1\right]}

\global\long\def\pr#1{\opstyle P \left[ #1 \right]}

\global\long\def\prcond#1#2{\opstyle P \left[\left. #1 \right\vert #2 \right]}

\global\long\def\excond#1#2{\opstyle E \left[\left. #1 \right\vert #2 \right]}

\global\long\def\excondls#1#2#3{\opstyle E _{#3}\left[\left. #1 \right\vert #2 \right]}

\global\long\def\exls#1#2{\opstyle{\opstyle E}_{#1}\left[ #2 \right]}

\global\long\def\prls#1#2{\opstyle P_{#1}\left[ #2 \right]}

\global\long\def\br#1{\left( #1 \right)}

\global\long\def\brq#1{\left[ #1 \right]}

\global\long\def\brw#1{\left\{  #1\right\}  }

\global\long\def\size#1{\left|#1\right|}

\global\long\def\setst#1#2{\left\{  #1\left|#2\right.\right\}  }

\global\long\def\setstcol#1#2{\left\{  #1:#2\right\}  }

\global\long\def\set#1{\left\{  #1\right\}  }

\global\long\def\adj#1{\delta\br{#1}}

\global\long\def\st#1{[#1] }

\global\long\def\rin#1{r^{in}\br{#1}}

\global\long\def\rout#1{r^{out}\br{#1}}

\global\long\def\rini#1#2{r_{#1}^{in}\br{#2}}

\global\long\def\routi#1#2{r_{#1}^{out}\br{#2}}

\global\long\def\Iin{{\cal I}^{in}}

\global\long\def\Iout{{\cal I}^{out}}

\global\long\def\kin{k^{in}}

\global\long\def\kout{k^{out}}

\global\long\def\Min#1{{\cal M}_{#1}^{in}}

\global\long\def\Mout#1{{\cal M}_{#1}^{out}}

\global\long\def\M{{\cal M}}

\global\long\def\deltaM{\delta^{\M}}

\global\long\def\betin#1{\beta_{#1}^{in}}

\global\long\def\betout#1{\beta_{#1}^{out}}

\global\long\def\Bin#1{B_{#1}^{in}}

\global\long\def\B#1{B_{#1}}

\global\long\def\Bout#1{B_{#1}^{out}}

\global\long\def\indi#1{\chi\brq{#1}}

\global\long\def\chr#1{\mathbf{1}_{#1}}

\global\long\def\pxE{\br{p_{e}x_{e}}_{e\in E}}

\global\long\def\pxEt#1{\br{p_{e}x_{e}}_{e\in E^{#1}}}

\global\long\def\xE{\br{x_{e}}_{e\in E}}

\global\long\def\xEt#1{\br{x_{e}}_{e\in E^{#1}}}

\global\long\def\e{\bar{e}}

\global\long\def\evalat#1#2{ #1 \Big|_{#2}}

\global\long\def\Df#1#2{\frac{\partial#1}{\partial#2}}

\global\long\def\event#1{\mathcal{E}_{#1}^{out}}
 \global\long\def\eventsucc#1{\mathcal{E}_{#1}^{in}}

\global\long\def\X#1{\hat{X}_{#1}}

\global\long\def\E{\hat{E}}

\maketitle

\begin{abstract}
In a \emph{stochastic probing }problem we are given a universe $E$,
where each element $e\in E$ is \emph{active} independently with probability
$p_{e}\in\left[ {0,1} \right]$, and only a \emph{probe} of $e$ can tell us
whether it is active or not. On this universe we execute a process
that one by one probes elements --- if a probed element is active,
then we have to include it in the solution, which we gradually construct.
Throughout the process we need to obey \emph{inner }constraints on
the set of elements taken into the solution, and \emph{outer} constraints
on the set of all probed elements. This abstract model was presented
by Gupta~and~Nagarajan~\cite{DBLP:conf/ipco/GuptaN13}, and provides
a unified view of a number of problems. Thus far all the results in
this general framework pertain only to the case in which we are maximizing
a linear objective function of the successfully probed elements. In
this paper we generalize the stochastic probing problem by considering
a monotone submodular objective function. We give a $(1-1/e)/(k^{in}+k^{out}+1)$-approximation
algorithm for the case in which we are given $k^{in}\ge0$ matroids
as inner constraints and $k^{out}\ge1$ matroids as outer constraints.
There are two main ingredients behind this result.
First is a previously unpublished stronger bound on the continuous greedy algorithm due to Vondrak~\cite{Vondrak:email}.
Second is a rounding procedure that also allows us to obtain an improved $1/(k^{in}+k^{out})$-approximation for linear objective functions. 
\end{abstract}

\section{Introduction}

\label{sec:introduction}

Uncertainty in input data is a common feature of most practical problems,
and research in finding good solutions (both experimental and theoretical)
for such problems has a long history dating back to 1950~\cite{Beale:convex,Dantzig:uncertainty}.
We consider adaptive stochastic optimization problems in the framework
of Dean~et~al.\ \cite{DBLP:journals/mor/DeanGV08}. Here the solution
is in fact a process, and the optimal one might even require larger
than polynomial space to describe. Since the work of Dean~et~al.\ a
number of such problems were introduced~\cite{DBLP:conf/icalp/ChenIKMR09,DBLP:conf/latin/GoemansV06,DBLP:conf/stoc/GuhaM07,DBLP:conf/soda/GuhaM07,DBLP:conf/wine/AsadpourNS08,DBLP:conf/focs/GuptaKMR11,DBLP:conf/soda/DeanGV05}.
Gupta~and~Nagarajan~\cite{DBLP:conf/ipco/GuptaN13} present an
abstract framework for a subclass of adaptive stochastic problems
giving a unified view for Stochastic Matching~\cite{DBLP:conf/icalp/ChenIKMR09}
and Sequential Posted Pricing~\cite{DBLP:conf/stoc/ChawlaHMS10}.

We describe the framework following~\cite{DBLP:conf/ipco/GuptaN13}.
We are given a universe $E$, where each element $e\in E$ is \emph{active
}with probability $p_{e}\in[0,1]$ independently. The only way to
find out if an element is active, is to \emph{probe }it. We call a
probe \emph{successful} if an element turns out to be active. On universe
$E$ we execute an algorithm that probes the elements one-by-one.
If an element is active, the algorithm must add it to the current
solution. In this way, the algorithm gradually constructs a solution
consisting of active elements.

Here, we consider the case in which we are given constraints on both
the elements probed and the elements included in the solution. Formally,
suppose that we are given two independence systems of downward-closed
sets: an \emph{outer} independence system $\br{E,\Iout}$ restricting
the set of elements probed by the algorithm, and an \emph{inner} independence
system $\br{E,\Iin}$, restricting the set of elements taken by the
algorithm. We denote by $Q^{t}$ the set of elements probed in the
first $t$ steps of the algorithm, and by $S^{t}$ the subset of active
elements from $Q^{t}$. Then, $S^{t}$ is the partial solution constructed
by the first $t$ steps of the algorithm. We require that, at each
time $t$, $Q^{t}\in\Iout$ and $S^{t}\in\Iin$. Thus, at each time
$t$, the element $e$ that we probe must satisfy both $Q^{t-1}\cup\{e\}\in\Iout$
and $S^{t-1}\cup\{e\}\in\Iin$. Gupta and Nagarajan \cite{DBLP:conf/ipco/GuptaN13}
considered many types of systems $\Iin$ and $\Iout$, but we focus
only on matroid intersections, i.e.\ on the special case in which
$\Iin$ is an intersection of $\kin$ matroids $\Min 1,\ldots,\Min{\kin}$,
and $\Iout$ is an intersection of $\kout$ matroids $\Mout 1,\ldots,\Mout{\kout}$.
We always assume that $\kout\geq1$ and $\kin\ge0$. We assume familiarity
with matroid algorithmics (see \cite{Schrijver:book}, for example)
and, above all, with principles of approximation algorithms
(see \cite{Vazirani}, for example).

Considering submodular objective functions is a common practice in
combinatorial optimization as it extends the range of applicability
of many methods. So far, the framework of stochastic probing has been
used to maximize the expected weight of the solution found by the
process. We were given weights $w_{e}\ge0$ for $e\in E$ and, if
$S$ denotes the solution at the end of a process, the goal was to
maximize $\exls S{\sum_{e\in S}w_{e}}$. We generalize the framework
as we consider a monotone submodular function $f:2^{E}\mapsto\mathbb{R}_{\geq0}$,
and objective of maximizing $\exls S{f\br S}$.

\subsection{Our results}

Our result is a new algorithm for stochastic probing problem based
on iterative randomized rounding of linear programs and the \emph{continuous
greedy process} introduced by Calinescu~et~al.~\cite{Calinescu2011}.

\begin{thm} \label{thm:submodular}An algorithm based on the continuous
greedy process and iterative randomized rounding is a $\frac{\br{1-e^{-1}}}{\kin+\kout+1}$-approximation
for stochastic probing problem with monotone submodular objective
function. \end{thm} Additionally, we improve the bound of $\frac{1}{4\br{\kin+\kout}}$
given by Gupta and Nagarajan~\cite{DBLP:conf/ipco/GuptaN13} in the
case of a linear objective. \begin{thm} \label{thm:linear}The iterative
randomized rounding algorithm is a $\frac{1}{\kin+\kout}$-approximation
for the stochastic probing problem with a linear objective function.
\end{thm}

\subsection{Applications}

\paragraph*{On-line dating and kidney exchange \cite{DBLP:conf/icalp/ChenIKMR09}}

Consider an online dating service. For each pair of users, machine
learning algorithms estimate the probability that they will form a
happy couple. However, only after a pair meets do we know for sure
if they were successfully matched (and together leave the dating service).
Users have individual patience numbers that bound how many unsuccessful
dates they are willing to go on until they will leave the dating service
forever. The objective of the service is to maximize the number of
successfully matched couples.

To model this as a stochastic probing problem, users are represented
as vertices $V$ of a graph $G=\br{V,E}$, where edges represent matched
couples. Set $E$ of edges is our universe on which we make probes,
with $p_{e}$ being the probability that a couple $e=\br{u_{1},u_{2}}$
forms a happy couple after a date. The inner constraints are matching
constraints --- a user can be in at most one couple ---, and outer
constraints are $b$--matching --- we can probe at most $t\br u$
edges adjacent to user $u$, where $t\br u$ denotes the patience
of $u$. Both inner and outer constraints are intersections of two
matroids for bipartite graphs. In similar way we can model kidney
exchanges.

In weighted bipartite case Theorem~\ref{thm:linear} gives a 1/4-approximation.
Even though $b$--matchings in general graphs are not intersections
of two matroids, we are able to exploit the matching structure to
give the same factor-1/4 approximation. Since the technique is very
similar to the case of intersection of two matroids, we omit the proof.
This matches the current-best bound for general graphs of Bansal~et~al.~\cite{Bansal:woes}, who also give a 1/3-approximation in the bipartite case.

\iffalse
This factor matches the current best bound of Bansal~et~al.~\cite{Bansal:woes}.

For weighted bipartite case without patience constraints our framework gives a 1/3-approximation.
Similarly as Bansal~et~al.~\cite{Bansal:woes}, we can use randomized pipage rounding technique~\cite{DBLP:journals/jacm/GandhiKPS06}
to get rid of the patience constraints, and get factor-1/3 approximation as well.
Details on how our framework can be combined with randomized pipage rounding
can be found in technical report~\cite{DBLP:journals/corr/AdamczykSW13}.
\fi

\iffalse
\begin{theorem} An algorithm based on iterative randomized rounding algorithm
is a 1/4-approximation for the weighted stochastic matching problem
in general graphs. \end{theorem} Even though our framework directly does
not yield a $1/3$-approximation in the bipartite case, we can combine
it with the randomized pipage rounding technique of Ghandi et al.~\cite{DBLP:journals/jacm/GandhiKPS06}
to obtain the following Theorem.

\begin{restatable}{theorem}{stochmatchbip} An algorithm based on
randomized pipage rounding and iterative randomized rounding algorithm
is a 1/3-approximation for the weighted stochastic matching problem
in bipartite graphs. \end{restatable}

Above factors of $1/4$ and $1/3$ match the current-best bounds of
Bansal~et~al.~\cite{Bansal:woes}.
\fi

\paragraph*{Bayesian mechanism design \cite{DBLP:conf/ipco/GuptaN13}}

Consider the following mechanism design problem. There are $n$ agents
and a single seller providing a certain service. Agent's $i$ value
for receiving service is $v_{i}$, drawn independently from a distribution
$D_{i}$ over set $\set{0,1,\ldots,B}$. The valuation $v_{i}$ is
private, but the distribution $D_{i}$ is known. The seller can provide
service only for a subset of agents that belongs to system ${\cal I}\in2^{\brq n}$,
which specifies feasibility constraints. A mechanism accepts bids
of agents, decides on subset of agents to serve, and sets individual
prices for the service. A mechanism is called truthful if agents bid
their true valuations. Myerson's theory of virtual valuations yields
\emph{truthful }mechanisms that maximize the expected revenue of a
seller, although they sometimes might be impractical. On the other
hand, practical mechanisms are often non-truthful. The Sequential
Posted Pricing Mechanism (SPM) introduced by Chawla~et~al.\ \cite{DBLP:conf/stoc/ChawlaHMS10}
gives a nice trade-off --- it is truthful, simple to implement, and
gives near-optimal revenue. An SPM offers each agent a ``take-it-or-leave-it''
price for the service. Since after a refusal a service won't be provided,
it is easy to see that an SPM is a truthful mechanism.

To see an SPM as a stochastic probing problem, we consider a universe
$E=\brq n\times\set{0,1,\ldots,B}$, where element $\br{i,c}$ represents
an offer of price $c$ to agent $i$. The probability that $i$ accepts
the offer is $\pr{v_{i}\geq c}$, and seller earns $c$ then. Obviously,
we can make only one offer to an agent, so outer constraints are given
by a partition matroid; making at most one probe per agent also overcomes
the problem that probes of $\br{i,1},...,\br{i,B}$ are not independent.
The inner constraints on universe $\brq n\times\set{0,1,\ldots,B}$
are simply induced by constraints ${\cal I}$ on $\brq n$.

Gupta and Nagarajan~\cite{DBLP:conf/ipco/GuptaN13} give an LP relaxation
for any single-seller Bayesian mechanism design problem. Provided
that we can optimize over ${\cal P}\br{{\cal I}}$, the LP can be
used to construct an efficient SPM. Moreover, the approximation guarantee
of the constructed SPM is with respect to the optimal mechanism, which
need not be an SPM.

In the case constraints ${\cal I}$ are an intersection of $k$ matroids
the resulting SPM is a $\frac{1}{4\br{k+1}}$-approximation~\cite{DBLP:conf/ipco/GuptaN13}.
Here, we give an improved approximation algorithm with a factor-$\frac{1}{k+1}$
guarantee. In particular, when $k=1$ we match~\cite{DBLP:conf/stoc/ChawlaHMS10,DBLP:conf/stoc/KleinbergW12}
with $1/2$-approximation.

\subsection{Related work}

\label{sec:related-work}

The stochastic matching problem with applications to online dating
and kidney exchange was introduced by Chen~et~al.\ \cite{DBLP:conf/icalp/ChenIKMR09},
where authors proved a 1/4-approximation of a greedy strategy for
unweighted case. The authors also show that the simple greedy approach
gives no constant approximation in the weighted case.
\iffalse that the set of \emph{probes}
must form a valid matching, which is special case of the general probing
setting with outer matroid constraints.\fi
Their bound was later improved
to 1/2 by Adamczyk~\cite{Adamczyk:greedy}. As noted in our discussion
of applications, Bansal~et~al.\ \cite{Bansal:woes} gave 1/3 and 1/4-approximations
for weighted stochastic matching in bipartite and general graphs, respectively.

Sequential Posted Pricing mechanisms were investigated first by Chawla~et~al.~\cite{DBLP:conf/stoc/ChawlaHMS10},
followed by Yan~\cite{DBLP:conf/soda/Yan11}, and Kleinberg~and~Weinberg~\cite{DBLP:conf/stoc/KleinbergW12}.
Gupta~and~Nagarajan~\cite{DBLP:conf/ipco/GuptaN13} were first
to propose looking at SPM from the point of view of stochastic adaptive
problems.

Asadpour~et~al.\ \cite{DBLP:conf/wine/AsadpourNS08} were first
to consider a stochastic adaptive problem with submodular objective
function. In our terms, they considered only a single outer matroid
constraint.

Work of Calinescu~et~al.\ \cite{Calinescu2011} \iffalse and Vondrák~\cite{Vondrak:PhD}\fi
provides the tools for submodular functions we use in this paper.
The method of~\cite{DBLP:conf/soda/Yan11} was based on ``correlation
gap''~\cite{DBLP:conf/soda/AgrawalDSY10}, something we address
implicitly in Subsection~\ref{sec:submodular-functions}. %We use there facts
%from~\cite{Vondrak:PhD}, some of which are also used in the analysis
%of a non-adaptive strategy by Asadpour~et~al.\ \cite{DBLP:conf/wine/AsadpourNS08}.

\section{Preliminaries}

\label{sec:preliminaries}

For set $S\subseteq E$ and element $e\in E$ we use $S+e$ to denote
$S\cup\{e\}$, and $S-e$ to denote $S\setminus\set e$. For set $S\subseteq E$
we shall denote by $\chr S$ a characteristic vector of set $S$,
and for a single element $e$ we shall write $\chr e$ instead of
$\chr{\set e}$. For random event ${\cal A}$ we shall denote by $\indi{{\cal A}}$
a 0-1 random variable that indicates whether ${\cal A}$ occurred.
The optimal strategy will be denoted by $OPT$, and we shall denote
the expected objective value of its outcome as $\ex{OPT}$.

\subsection{Matroids and polytopes}

\label{sec:matroids-polytopes}

Let ${\cal M}=\br{E,{\cal I}}$ be a matroid, where $E$ is the universe
of elements and ${\cal I}\subseteq2^{E}$ is a family of independent
sets. For element $e\in E$, we shall denote the matroid ${\cal M}$
with $e$ contracted by ${\cal M}/e$, i.e.\ ${\cal M}/e=\br{E-e,\setst{S\subseteq E-e}{S+e\in{\cal I}}}$.

The following lemma is a slightly modified%
\footnote{The difference is that we do not assume that $A,B$ are bases, but
independent sets of the same size.%
} basis exchange lemma, which can be found in~\cite{Schrijver:book}.
\begin{lemma} \label{lem:exchange}Let $A,B\in{\cal I}$ and $\size A=\size B$.
There exists a bijection $\phi:A\mapsto B$ such that: 1) $\phi\br e=e$
for every $e\in A\cap B$, 2) $B-\phi\br e+e\in{\cal I}$. \end{lemma}

We shall use the following corollary, where we consider independent
sets of possibly different sizes.

\begin{restatable}{corollary}{matroidcor}
\label{cor:mapping}Let $A,B\in{\cal I}$. We can find assignment $\phi_{A,B}:A\mapsto B\cup\{\bot\}$ such that: 
\begin{enumerate}
\item $\phi_{A,B}\br e=e$ for every $e\in A\cap B$, 
\item for each $f\in B$ there exists at most one $e\in A$ for which $\phi_{A,B}\br e=f$, 
\item for $e\in A\setminus B$, if $\phi_{A,B}\br e=\bot$ then $B+e\in{\cal I}$,
otherwise $B-\phi_{A,B}\br e+e\in{\cal I}$. 
\end{enumerate}
\end{restatable}

We consider optimization over \emph{matroid polytopes } which have
the general form ${\cal P}\br{{\cal M}}=\setst{x\in\mathbb{R}_{\geq0}^{E}}{\forall_{A\in{\cal I}}\sum_{e\in A}x_{e}\leq r_{{\cal M}}\br A}$,
where $r_{{\cal M}}$ is the rank function of ${\cal M}$. We know~\cite{Schrijver:book}
that the matroid polytope ${\cal P}\br{{\cal M}}$ is equivalent to
the convex hull of $\setst{\chr A}{A\in{\cal I}}$, i.e.\ characteristic
vectors of all independent sets of ${\cal M}$. Thus, we can represent
any $x\in{\cal P}\br{{\cal M}}$ as $x=\sum_{i=1}^{m}\beta_{i}\cdot\chr{B_{i}}$,
where $B_{1},\ldots,B_{m}\in{\cal I}$ and $\beta_{1},\ldots,\beta_{m}$
are non-negative weights such that $\sum_{i=1}^{m}\beta_{i}=1$ .
We shall call sets $B_{1},\ldots,B_{m}$ a \emph{support }of $x$
in ${\cal P}\br{{\cal M}}$.

\subsection{Submodular functions}

\subsubsection{Multilinear extension}

A set function $f:2^{E}\mapsto\mathbb{R}_{\ge0}$ is \emph{submodular},
if for any two subsets $S,T\subseteq E$ we have $f\br{S\cup T}+f\br{S\cap T}\leq f\br S+f\br T$.
We call function $f$ \emph{monotone}, if for any two subsets $S\subseteq T\subseteq E:f\br S\leq f\br T$.
For a set $S\subseteq E$, we let $f_{S}(A)=f(A\cup S)-f(S)$ denote
the marginal increase in $f$ when the set $A$ is added to $S$.
Note that if $f$ is monotone submodular, then so is $f_{S}$ for
all $S\subseteq E$. Moreover, we have $f_{S}(\emptyset)=0$ for all
$S\subseteq E$, so $f_{S}$ is normalized. Without loss of generality,
we assume also that $f\br{\emptyset}=0$.

We consider the \emph{multilinear extension} $F:[0,1]^{E}\mapsto\mathbb{R}_{\ge0}$
of $f$, whose value at a point $y\in\brq{0,1}^{E}$ is given by 
\[
F(y)=\sum_{A\subseteq E}f(A)\prod_{e\in A}y_{e}\prod_{e\not\in A}(1-y_{e}).
\]
 Note that $F\br{\chr A}=f\br A$ for any set $A\subseteq E$, so
$F$ is an extension of $f$ from discrete domain $2^{E}$ into a
real domain $\brq{0,1}^{E}$. The value $F(y)$ can be interpreted
as the expected value of $f$ on a random subset $A\subseteq E$ that
is constructed by taking each element $e\in E$ with probability $y_{e}$.
Following this interpretation, Calinescu~et~al.~\cite{Calinescu2011}
show that $F(y)$ can be estimated to any desired accuracy in polynomial
time, using a sampling procedure.

Additionally, they show that $F$ has the following properties, which
we shall make use of in our analysis:

\begin{lemma} \label{lem:multilinear}The multilinear extension $F$
is linear along the coordinates, i.e.\ for any point $x\in\brq{0,1}^{E}$,
any element $e\in E$, and any $\xi\in\brq{-1,1}$ such that $x+\xi\cdot\chr e\in\brq{0,1}^{E}$,
it holds that $F\br{x+\xi\cdot\chr e}-F\br x=\xi\cdot\Df F{y_{e}}\br x$,
where $\Df F{y_{e}}\br x$ is the partial derivative of $F$ in direction
$y_{e}$ at point $x$. \end{lemma}

\begin{lemma} \label{lem:partialF}If $F:\brq{0,1}^{E}\mapsto\mathbb{R}$
is a multilinear extension of monotone submodular function $f:2^{E}\mapsto\mathbb{R}$,
then 1) function $F$ has second partial derivatives everywhere; 2)
for each $e\in E$, $\Df F{y_{e}}\geq0$ everywhere; 3) for any $e_{1},e_{2}\in E$
(possibly equal), $\frac{\partial^{2}F}{\partial y_{e_{1}}\partial y_{e_{2}}}\leq0$,
which means that $\Df F{y_{e_{2}}}$ is non-increasing with respect
to $y_{e_{1}}$. \end{lemma}

\subsubsection{Continuous greedy algorithm}
\label{sec:submodular-functions}

\label{sub:continousgreedy}

In~\cite{Calinescu2011} the authors utilized the multilinear extension
in order to maximize a submodular monotone function over a matroid
constraint. They showed that a \emph{continuous greedy algorithm}
finds a $\br{1-1/e}$-approximate maximum of the above extension $F$
over any downward closed polytope. In the special case of the matroid
polytope, they show how to employ the pipage rounding~\cite{DBLP:journals/jco/AgeevS04}
technique to the fractional solution to obtain an integral solution.

Another extension of $f$ studied in \cite{DBLP:conf/ipco/CalinescuCPV07}
is given by: 
\[
f^{+}(y)=\max\setst{\sum_{A\subseteq E}\alpha_{A}f(A)}{\sum_{A\subseteq E}\alpha_{A}\le1,\ \forall A\subseteq E:\alpha_{A}\ge0,\ \forall j\in E:\sum_{A:j\in A}\alpha_{A}\le y_{j}}
\]
 Intuitively, the solution $(\alpha_{A})_{A\subseteq E}$ above represents
the distribution over $2^{E}$ that maximizes the value $\ex{f(A)}$
subject to the constraint that its marginal values satisfy $\pr{i\in A}\le y_{i}$.
The value $f^{+}(y)$ is then the expected value of $\ex{f(A)}$ under
this distribution, while the value of $F(y)$ is the value of $\ex{f(A)}$
under the particular distribution that places each element $i$ in
$A$ independently. However, the following allows us to relate the
value of $F$ on the solution of the continuous greedy algorithm to
the optimal value of the relaxation $f^{+}$. \begin{restatable}{lemma}{vondrak}
\label{lem:continuous-greedy} Let $f$ be a submodular function with
multilinear extension $F$, and let $\mathcal{P}$ be any downward
closed polytope. Then, the solution $x\in\mathcal{P}$ produced by
the continuous greedy algorithm satisfies $F(x)\ge(1-1/e)\max_{y\in\mathcal{P}}f^{+}(y)$.
\end{restatable} This follows from a simple modification of the continuous
greedy analysis, given by Vondrák \cite{Vondrak:email}.

\subsection{Overview of the iterative randomized rounding approach}

\label{sec:overv-iter-round}

We now give a description of the general rounding approach that we
employ in both the linear and submodular case. In each case, we formulate
a mathematical programming relaxation of the following general form
\begin{equation}
\max_{x\in\brq{0,1}^{E}}\setst{g(x)}{\forall j\in\brq{\kin}:p\cdot x\in{\cal P}\br{\Min j};\,\forall j\in\brq{\kout}:x\in{\cal P}\br{\Mout j}}\label{eq:general-relaxation}
\end{equation}
 with $p\in\brq{0,1}^{E}$ being the vector of probabilities. Here
$g:[0,1]^{E}\mapsto\mathbb{R}_{\ge0}$ is an objective function chosen
so that the optimal value of \eqref{eq:general-relaxation} can be
used to bound the expected value of an optimal policy for the given
instance using the following lemma. Note that our program will always
have constraints as given in \eqref{eq:general-relaxation}, only
the objective function $g$ changes between the linear and monotone
submodular cases.

\begin{lemma} \label{lem:math-programming-bound} Let $OPT$ be the
optimal feasible strategy for some stochastic probing problem in our
general setting, and define $x_{e}=\pr{OPT\mbox{ probes }e}$. Then,
$x=(x_{e})_{e\in E}$ is a feasible solution to the related relaxation
of the form \eqref{eq:general-relaxation}. \end{lemma} \begin{proof}
Since $OPT$ is a feasible strategy, the set of elements $Q$ probed
by any execution of $OPT$ is always an independent set of each outer
matroid ${\cal M}=\br{E,\Iout_{j}}$, i.e.\ $\forall_{j\in\brq{\kout}}Q\in\Iout_{j}$.
Thus, for any $j\in\brq{\kout}$, the vector $\ex{\chr Q}=x$ may
be represented as a convex combination of vectors from $\setst{\chr A}{A\in\Iout_{j}}$,
and hence $x\in{\cal P}\br{\Mout j}$. Analogously, the set of elements
$S$ that were successfully probed by $OPT$ satisfy $\forall_{j\in\brq{\kin}}S\in\Iin_{j}$
for every possible execution of $OPT$. Hence, for any $j\in\brq{\kin}$
the vector $\ex{\chr S}=p\cdot x$ may be represented as a convex
combination of vectors from $\setst{\chr A}{A\in\Iin_{j}}$, and hence
$x\in{\cal P}\br{\Min j}$. \end{proof}

Suppose that $f$ is the objective function for a given instance of
stochastic probing over a universe $E$ of elements. Our algorithm
first obtains a solution $x^{0}$ to a relaxation of the form \eqref{eq:general-relaxation}
using either linear programming or the continuous greedy algorithm.
Our algorithm proceeds iteratively, maintaining a current set of constraints,
a current fractional solution $x$, and a current set $S$ of elements
that have been successfully probed. Initially, the constraints are
as given in \eqref{eq:general-relaxation}, $x=x^{0}$, and $S=\emptyset$.
At each step, the algorithm selects single element $\e$ to probe,
then permanently sets $x_{\e}$ to 0. It then updates the outer constraints,
replacing $\Mout j$ with $\Mout j/\e$ for each $j\in[\kout]$. If
the probe succeeds, the algorithm adds $\e$ to $S$ and then updates
the inner constraints, replacing $\Min j$ with $\Min j/\e$ for each
$j\in[\kin]$. Finally, we modify our fractional solution $x$ so
that it is feasible for the updated constraints. The algorithm terminates
when the current solution $x=0^{E}$.

In order to analyze the approximation performance of our algorithm,
we keep track of a current potential value $z$, related to the value of the
remaining fractional solution $x$.  Let $x^{t}$, $z^{t}$, and $S^{t}$ be the current value of $x$, $z$, and $S$
at the beginning of step $t+1$. We show that initially we have $z^{0}=g\br{x^{0}}\ge\beta\cdot\ex{OPT}$
for some constant $\beta\in[0,1]$, and then analyze the expected
decrease $z^{t}-z^{t+1}$ at an arbitrary step $t+1$. We show that
for each step we have $\alpha\cdot\ex{z^{t}-z^{t+1}}\le\ex{f(S^{t+1})-f(S^{t})}$,
for some $\alpha<1$. That is, the expected increase in the value
of the current solution is at least $\alpha$ times the expected decrease
in $z$. Then, we employ the following Lemma to conclude that the
algorithm is an $\alpha\beta$-approximation in expectation. The proof
is based on Doob's optional stopping theorem for martingales. Hence,
we need to deploy language from martingale theory, such as stopping
time and filtration.  See~\cite{probwithmartin} for extended background
on martingale theory.

\begin{restatable}{lemma}{martingale} \label{lem:martingale} Suppose
the algorithm runs for $\tau$ steps and that $z^{0}=g\br{x^{0}}\ge\beta\cdot\ex{OPT}$,
$z^{\tau}=0$. Let $(\mathcal{F}_{t})_{t\ge0}$ be the filtration
associated with our iterative algorithm, where $\mathcal{F}_{i}$
represents all information available after the $i$th iteration. Finally,
suppose that in each step in our iterative rounding procedure, $\excond{f(S^{t+1})-f(S^{t})}{\mathcal{F}_{t}}\ge\alpha\cdot\excond{z^{t}-z^{t+1}}{\mathcal{F}_{t}}$.
Then, the final solution $S^{\tau}$ produced by the algorithm satisfies
$\ex{f(S^{\tau})}\ge\alpha\beta\cdot\ex{OPT}$. \end{restatable}

\begin{proof} Let $G_{t+1}$ be the gain $f(S^{t+1})-f(S^{t})$ in
$f$ at step $t+1$, and let $L_{t+1}$ be the corresponding loss
$z^{t}-z^{t+1}$ in $z$ at time $t+1$. We set $G_{0}=L_{0}=0$.
Define variable $D_{t}=G_{t}-\alpha\cdot L_{t}$. The sequence of
random variables $\br{D_{0}+D_{1}+...+D_{t}}_{t\geq0}$ forms a sub-martingale,
i.e.\ 
\[
\excond{\sum_{i=0}^{t+1}D_{i}}{{\cal F}_{t}}=\sum_{i=0}^{t}D_{i}+\excond{G_{t+1}-\alpha\cdot L_{t+1}}{{\cal F}_{t}}\geq\sum_{i=0}^{t}D_{i}.
\]
 Let $\tau$ be the step in which the algorithm terminates, i.e.~$\tau=\min\setst t{x^{t}=0^{E}}$.
Then, the event $\tau=t$ depends only on $\mathcal{F}_{0},\ldots,\mathcal{F}_{t}$,
so $\tau$ is a stopping time. Also, by the definition of the algorithm
$x^{\tau}=0^{E}$. It is easy to verify that all the assumptions of
Doob's optional stopping theorem are satisfied, and from this theorem
we get that $\ex{\sum_{i=0}^{\tau}D_{i}}\ge\ex{D_{0}}$. Since $D_{0}=0$,
we have 
\[
0\le\ex{\sum_{i=0}^{\tau}D_{i}}=\ex{\sum_{i=0}^{\tau}G_{i}-\alpha\cdot\sum_{i=0}^{\tau}L_{i}}=\ex{\sum_{i=0}^{\tau}G_{i}}-\alpha\cdot\ex{\sum_{i=0}^{\tau}L_{i}}.
\]
 It remains to note that $\sum_{i=0}^{\tau}G_{i}=f\br{S^{\tau}}$
is the total gain of the algorithm, so $\ex{\sum_{i=0}^{\tau}G_{i}}=\ex{f(S^{\tau})}$.
On the other hand, $\sum_{i=0}^{\tau}L_{i}=g(x^{0})-g(x^{\tau})=g(x^{0})\ge\beta\cdot\ex{OPT}$.
\end{proof} Henceforth, we will implicitly condition on all information
$\mathcal{F}_{t}$ available to the algorithm just before it makes
step $t+1$. That is, when discussing step $t+1$ of the algorithm,
we write shortly $\ex{\cdot}$ instead of $\excond{\cdot}{\mathcal{F}_{t}}$.

\section{Linear stochastic probing}

In this setting, we are given a weight $w_{e}$ and a probability
$p_{e}$ for each element $e\in E$ and $f(S)$ is simply $\sum_{e\in S}w_{e}$.
We consider the relaxation \eqref{eq:general-relaxation} in which
$g(x)=f(x)$. Then, Lemma \ref{lem:math-programming-bound} shows
that the optimal policy $OPT$ must correspond to some feasible solution
$x^{*}$ of \eqref{eq:general-relaxation}. Moreover, because $f$
is linear, $\ex{OPT}=\sum_{e\in S}\pr{OPT\mbox{ probes e}}p_{e}w_{e}=\sum_{e\in S}x_{e}p_{e}w_{e}=f(x^{*})$.

At each step, our algorithm randomly selects an element $\e$ to probe.
Let $\Sigma=\sum_{e\in E}x_{e}$ Then, our algorithm chooses element
$e$ with probability $x_{e}/\Sigma$. As discussed in the previous
overview, it then sets $x_{\e}=0$ and carries out the probe, updating
the matroid constraints to reflect both the choice of $\e$ and the
probe. Finally, it updates $x$ to obtain a new fractional solution
that is feasible in the updated constraints. Note that because $x_{e}$
is set to 0 after probing $e$, we will never probe an element $e$
twice.

Let us now describe how to update the current solution $x$ to ensure
feasibility in each of the updated matroid constraints. Let $\e$
be the element that we probed and let $\Mout j$ be some outer matroid.
Currently we have $x\in\mathcal{P}(\Mout j)$ and we must obtain a
solution $x'$ so that $x'\in\mathcal{P}(\Mout j/\e)$. We represent
the vector $x$ as a convex combination of independent sets $x=\sum_{i=1}^{m}\betout i\chr{\Bout i}$,
where $\Bout 1,\ldots,\Bout m$ is the support of $x$ with respect
to matroid $\Mout j$. We obtain $x'\in\mathcal{P}(\Mout j/\e)$ by
replacing each independent set $\Bout b$ for which $\Bout b+\e\not\in\Mout j$
with some other set $\Bout c$ such that $\Bout c+\e\in\Mout j$.
We pick one set $\Bout a$ with $\e\in\Bout a$ to guide the update
process. We pick the set $\Bout a\ni\e$ at random with probability
$\betout a/x_{\e}$ (note that for any element $e$, $\sum_{a:e\in\Bout a}\betout a=x_{e}$).
For any set $\Bout b:\e\notin\Bout b$, let $\phi_{a,b}$ be the mapping
from $\Bout a$ into $\Bout b$ from Corollary~\ref{cor:mapping}.
If $\phi_{a,b}\br{\e}=\perp$, or $\phi_{a,b}(\e)=\e$, then in fact
$\Bout b+\e\in\Mout j$, and we can just include $\Bout b$ in the
support of $\Mout j/\e$. Otherwise, we substitute $\Bout b$ with
$\Bout b-\phi_{a,b}\br{\e}$ in the support of $\xE$ in ${\cal P}\br{\Mout j/\e}$,
since we know that $\Bout b-\phi_{a,b}\br{\e}+\e\in\Mout j$.

Similarly, if $\e$ is successfully probed we must perform a support
update for each inner matroid. Here, we proceed as in the case of
the outer matroids, except we have $p\cdot x\in\Min j$ and must obtain
$x'$ such that $p\cdot x'\in\Min j/e$. We write $p\cdot x$ as a
combination independent sets $p\cdot x=\sum_{i=1}^{m}\betin i\chr{\Bin i}$,
and now choose a random set $\Bin a\ni\e$ to guide the support update
with probability $\betin a/p_{\e}x_{\e}$. (note that for any element
$e$, we have $\sum_{a:e\in\Bin a}\betin a=p_{e}x_{e}$). As in the
previous case, we replace $\Bin b$ with $\Bin b-\phi_{a,b}(\e)$
for each base $\Bin b$ such that $\Bin b+\e\not\in\Min j$.

We now turn to the analysis of the probing algorithm. Suppose that
the algorithm runs for $\tau$ steps and consider the quantity $z^{t}=f(x^{t})$.
Then, $z^{0}=f(x^{0})\ge\ex{OPT}$ and $z^{\tau}=f(0^{E})=0$, so
the conditions of Lemma \ref{lem:martingale} are satisfied with $\beta=1$.
It remains to bound the expected loss $\ex{z^{t}-z^{t+1}}$ in step
$t+1$. In order to do this, we consider the value $\delta_{i}=p_{i}(x_{i}^{t}-x_{i}^{t+1})$
for each $i\in E$. We consider arbitrary step $t+1$, but we are
going to denote $x^{t}$ by $x$ and $x^{t+1}$ by $x'$. The decrease
$\delta_{i}$ may be caused both by the probing step, in which we
set $x_{\e}'$ to 0, or by the matroid update step, in which we decrease
several coordinates of $x$. Let us first consider the losses due
to each matroid update.

\begin{lemma} \label{lem:outerdrop}Let $x$ and $x'$ be the current
fractional solution before and after one update for a given outer
matroid $\Mout j$. Then, for each $i\in E$, we have $\ex{\delta_{i}^{out}}\triangleq\ex{p_{i}(x_{i}-x'_{i})}\le\frac{1}{\Sigma}\br{1-x_{i}}p_{i}x_{i}$.
\end{lemma} \begin{proof} The expectation $\ex{\delta_{i}^{out}}$
is over the random choice of an element $\e$ to probe and the random
choice of an independent set to guide the update. Let $\event a$
denote the event that some set $\Bout a$ is chosen to guide a support
update for $\Mout j$.

In a given step the probability that the set $\Bout a$ is chosen
to guide the support update is equal to 
\begin{equation}
\pr{\event a}=\sum_{e\in\Bout a}\frac{x_{e}}{\Sigma}\frac{\betout a}{x_{e}}=\sum_{e\in\Bout a}\frac{\betout a}{\Sigma}=\size{\Bout a}\frac{\betout a}{\Sigma}.\label{eq:fact1}
\end{equation}
 Moreover, conditioned on the fact $\Bout a$ was chosen, the probability
that an element $e\in\Bout a$ was probed is uniform over the elements
of $\Bout a$: 
\begin{equation}
\prcond{e\mbox{ probed }}{\event a}=\pr{e\mbox{ probed}\wedge\event a}\left/\pr{\event a}\right.=\frac{x_{e}}{\Sigma}\frac{\betout a}{x_{e}}\left/\size{\Bout a}\frac{\betout a}{\Sigma}\right.=\frac{1}{\size{\Bout a}}.\label{eq:fact2}
\end{equation}
 We can write the expected decrease as $\ex{\delta_{i}^{out}}=\sum_{a=1}^{m}\pr{\event a}\cdot\excond{\delta_{i}^{out}}{\Bout a}$.
Note that for all $i\in\Bout a$, we have $\phi_{a,b}(i)=i$ for every
set $\Bout b\ni i$. Thus, the support update will not change $x_{i}$
for any $i\in\Bout a$, and so $\sum_{a=1}^{m}\pr{\event a}\cdot\excond{\delta_{i}^{out}}{\event a}=\sum_{a:i\notin\Bout a}\pr{\event a}\cdot\excond{\delta_{i}^{out}}{\event a}.$

Now let us condition on taking $\Bout a$ to guide the support update.
Consider a set $\Bout b\ni e$. If we remove $i$ from $\Bout b$,
and hence decrease $p_{i}x_{i}$ by $p_{i}\betout b$, it must be
the case that we have chosen to probe the single element $\phi_{a,b}^{-1}(i)\in\Bout a$.
The probability that we probe this element is $\frac{1}{\size{\Bout a}}$.
Hence 
\begin{align*}
 & \sum_{a:i\notin\Bout a}\pr{\event a}\cdot\excond{\delta_{i}^{out}}{\Bout a}\\
 & =\sum_{a:i\notin\Bout a}\pr{\event a}\cdot\br{\sum_{b:i\in\Bout b}p_{i}\betout b\cdot\pr{\phi_{a,b}^{-1}(i)\mbox{ is probed}\left|\event a\right.}}\\
 & \leq\sum_{a:i\notin\Bout a}\pr{\event a}\cdot\br{\sum_{b:i\in\Bout b}p_{i}\betout b\cdot\frac{1}{\size{\Bout a}}}\\
 & =\sum_{a:i\notin\Bout a}\pr{\event a}\cdot\frac{p_{i}x_{i}}{\size{\Bout a}}\\
 & =\sum_{a:i\notin\Bout a}\size{\Bout a}\frac{1}{\Sigma}\betout a\cdot\frac{p_{i}x_{i}}{\size{\Bout a}}\quad=\quad\frac{1}{\Sigma}\sum_{a:i\notin\Bout a}\betout ap_{i}x_{i}\quad=\quad\frac{1}{\Sigma}\br{1-x_{i}}p_{i}x_{i}.\quad\qedhere
\end{align*}
 \end{proof} \begin{lemma} \label{lem:innerdrop}Let $x$ be the
current fractional solution before and after one update for a given
inner matroid $\Min j$. Then, for each $i\in E$, we have $\ex{\delta_{i}^{in}}\triangleq\ex{p_{i}(x_{i}-x'_{i})}\le\frac{1}{\Sigma}\br{1-p_{i}x_{i}}p_{i}x_{i}$.
\end{lemma} \begin{proof} Because we only perform a support update
when the probe of a chosen element is successful, the expectation
$\ex{\delta_{i}^{in}}$ is over the random result of the probe, as
well as the random choice of element $\e$ to probe and the random
choice of a base to guide the update. We proceed as in the case of
Lemma \ref{lem:outerdrop}, now letting $\eventsucc a$ denote the
event that the probe was successful and $\Bin a$ is chosen to guide
the support update. We have: 
\begin{align*}
 & \pr{\eventsucc a}=\sum_{e\in\Bin a}p_{e}\frac{x_{e}}{\Sigma}\frac{\betin a}{p_{e}x_{e}}=\sum_{e\in\Bin a}\frac{\betin a}{\Sigma}=\size{\Bin a}\frac{\betin a}{\Sigma},\\
 & \prcond{e\mbox{ probed }}{\eventsucc a}=\pr{e\mbox{ probed}\wedge\eventsucc a}\left/\pr{\eventsucc a}\right.=p_{e}\frac{x_{e}}{\Sigma}\frac{\betin a}{p_{e}x_{e}}\left/\size{\Bin a}\frac{\betin a}{\Sigma}\right.=\frac{1}{\size{\Bin a}}.
\end{align*}
 By a similar argument as in Lemma \ref{lem:outerdrop} we then have
that $\ex{\delta_{i}^{in}}$ is at most: 
\begin{multline*}
\sum_{a:i\notin\Bin a}\pr{\eventsucc a}\cdot\br{\sum_{b:i\in\Bin b}\betin b\cdot\frac{1}{\size{\Bin a}}}=\sum_{a:i\notin\Bin a}\pr{\eventsucc a}\cdot\frac{p_{i}x_{i}}{\size{\Bin a}}\\
=\sum_{a:i\notin\Bout a}\size{\Bin a}\frac{1}{\Sigma}\betin a\cdot\frac{p_{i}x_{i}}{\size{\Bin a}}=\frac{1}{\Sigma}\sum_{a:i\notin\Bin a}\betin ap_{i}x_{i}=\frac{1}{\Sigma}(1-p_{i}x_{i})p_{i}x_{i}.\qedhere
\end{multline*}
 \end{proof} We perform the matroid updates sequentially for each
of the $\kin$ and $\kout$ matroids. Note that once we decrease a
coordinate $x_{i}$ to 0, it cannot be altered in any further updates,
so no coordinate is ever decreased below 0. Now, we consider the expected
decrease $\ex{\delta_{i}}=\ex{p_{i}(x_{i}-x_{i}')}$ due to both the
initial probing step, in which we decrease the probed element's coordinate
to 0, and the following matroid updates. We have: 
\begin{align}
\ex{\delta_{i}} & \le\pr{i\mbox{ probed}}p_{i}x_{i}+\kout\ex{\delta_{i}^{out}}+\kin\ex{\delta_{i}^{in}}\nonumber \\
 & =\frac{1}{\Sigma}p_{i}x_{i}^{2}+\kout\frac{1}{\Sigma}(1-x_{i})p_{i}x_{i}+\kin\frac{1}{\Sigma}(1-p_{i}x_{i})p_{i}x_{i}\nonumber \\
 & =\frac{1}{\Sigma}\kout p_{i}x_{i}-\frac{1}{\Sigma}(\kout-1)p_{i}x_{i}^{2}+\frac{1}{\Sigma}\kin p_{i}x_{i}-\frac{1}{\Sigma}\kin p_{i}^{2}x_{i}^{2}\nonumber \\
 & \leq\frac{\kout+\kin}{\Sigma}p_{i}x_{i}.\label{eq:linearloss}
\end{align}
 Because $z^{t}$ is a linear function of $x^{t}$, the expected total
decrease of $z$ in this step is then 
\[
\ex{z^{t}-z^{t+1}}=\sum_{i}\ex{\delta_{i}}w_{i}\le\frac{\kout+\kin}{\Sigma}\sum_{i}p_{i}x_{i}w_{i}.
\]
 On the other hand, the expected gain in $f(S)$ is $\sum_{e\in E}\pr{e\mbox{ probed}}p_{e}w_{e}=\frac{1}{\Sigma}\sum_{e\in E}w_{e}p_{e}x_{e}$.
Thus, by Lemma \ref{lem:martingale} the final solution $S^{\tau}$
produced by the algorithm satisfies $\ex{f(S^{\tau})}\ge\frac{1}{\kout+\kin}\ex{OPT}$.

\section{Submodular stochastic probing}

\label{sec:submodular-case}

We now consider the case in which we are given a set of elements $E$
each becoming active with probability $p_{e}$, and we seek to maximize
a given submodular function $f:2^{E}\mapsto\mathbb{R}_{\geq0}$. In
this case, we consider the relaxation \eqref{eq:general-relaxation}
in which $g(x)=f^{+}(p\cdot x)$. Then, Lemma \ref{lem:math-programming-bound}
shows that the optimal policy $OPT$ must correspond to some feasible
solution $x^{*}$ of \eqref{eq:general-relaxation}, where $x_{e}^{*}=\pr{OPT\mbox{ probes }e}$,
and hence $\pr{OPT\mbox{ takes }e}=p_{e}x_{e}^{*}$. The function
$f^{+}(p\cdot x^{*})$ gives the maximum value of $\exls{S\sim\mathcal{D}}{f(S)}$
over all distributions $\mathcal{D}$ satisfying $\prls{S\sim\mathcal{D}}{e\in S}=x_{e}^{*}p_{e}$.
Thus, $f^{+}(p\cdot x^{*})\ge\ex{OPT}$.

In general, we cannot obtain an optimal solution to this relaxation.
Instead, we apply the continuous greedy algorithm to a variant of
\eqref{eq:general-relaxation} in which $g(x)$ is given by $F(p\cdot x)$
to obtain an initial solution $x^{0}$. From Lemma \ref{lem:continuous-greedy}
we then have $F(p\cdot x^{0})\ge(1-1/e)f^{+}(p\cdot x^{*})\ge(1-1/e)\ex{f(OPT)}$.

Given $x^{0}$, our algorithm is exactly the same as in the linear
case. However, we must be more careful in our analysis. We define
the quantity 
\[
z^{t}=F(\chr{S^{t}}+p\cdot x^{t})-F(\chr{S^{t}})
\]
 where $S^{t}$ and $x^{t}$ are, respectively, the set of successfully
probed elements and the current fractional solution at time $t$.
Note that because after probing an element we set its variable to
zero, for all elements $i\in S$ we have $x_{i}=0$, and so indeed
$\chr{S^{t}}+p\cdot x^{t}\in[0,1]^{E}$. Suppose that the algorithm
runs for $\tau$ iterations, and note that $z^{0}=F(p\cdot x^{0})\ge(1-1/e)\ex{f(OPT)}$
and $z^{\tau}=F(\chr{S^{\tau}+p\cdot0^{E}})-F(S^{\tau})=0$, so the
conditions of Lemma \ref{lem:martingale} are satisfied with $\beta=(1-1/e)$.

We now analyze the expected decrease $z^{t}-z^{t+1}$ due to step
$t+1$ of the algorithm. Suppose that the algorithm selects element
$i$ to probe. Then, we have $S^{t+1}=S^{t}+i$ with probability $p_{i}$
and $S^{t+1}=S^{t}$ otherwise. Thus, we have 
\begin{align}
\ex{z^{t}-z^{t+1}}&=\ex{F(\chr{S^{t}}+x^{t}\cdot p)-F(\chr{S^{t}})}-\ex{F(\chr{S^{t+1}}+x^{t+1}\cdot p)-F(\chr{S^{t+1}})}\notag \\
&=\ex{F(\chr{S^{t+1}})-F(\chr{S^{t}})}+\ex{F(\chr{S^{t}}+x^{t}\cdot p)-F(\chr{S^{t+1}}+x^{t+1}\cdot p)} \notag \\
&\le\ex{F(\chr{S^{t+1}})-F(\chr{S^{t}})}+\ex{F(\chr{S^{t}}+x^{t}\cdot p)-F(\chr{S^{t}}+x^{t+1}\cdot p)}
\label{eq:big-nasty-expectation}
\end{align}
where in the last line, we have used the fact that $S^{t+1} \ge S^{t}$ and $F$ is increasing in all directions (Lemma \ref{lem:partialF}).
We shall first bound the second expectation in \eqref{eq:big-nasty-expectation}.  We consider the vector $\delta$ of decreases in $x$, given by $\delta=(x^{t}-x^{t+1})\cdot p$.
For each $i\in E$, let $w_{i}=\Df F{x_{i}}(\chr{S^{t}})=F(\chr{S^{t}+i})-F(\chr{S^{t}})$.
Let $y=\chr{S^{t}}+x^{t}\cdot p$, and suppose that we decrease the
coordinates of~$y$ one at a time to obtain $y-\delta=\chr{S^{t}}+x^{t+1}\cdot p$,
letting $y^{i}$ be the value of $y$ after the first $i-1$ coordinates
have been decreased.%
\footnote{With slight abuse of notation, we write $x_{i}$ for the value of
the $i$th decreased coordinate of $x$ and $\chr i$ for the characteristic
vector of this coordinate. That is, we identify an element with its
index%
} We then have: 
\iffalse
\begin{align*}
F(y)-F(y-\delta) & =\sum_{i}F(y^{i})-F(y^{i+1})\\
 & =\sum_{i}F(y^{i})-F(y^{i}-\delta_{i}\chr i)\\
 & =\sum_{i}\delta_{i}\Df F{x_{i}}(y^{i}-\delta_{i}\chr i)\\
 & \le\sum_{i}\delta_{i}\Df F{x_{i}}(\chr{S^{t}})\\
 & =\sum_{i}\delta_{i}w_{i},
\end{align*}
\fi
\begin{multline*}
F(y)-F(y-\delta)  =\sum_{i}F(y^{i})-F(y^{i+1})  =\sum_{i}F(y^{i})-F(y^{i}-\delta_{i}\chr i)\\
  =\sum_{i}\delta_{i}\Df F{x_{i}}(y^{i}-\delta_{i}\chr i) \le\sum_{i}\delta_{i}\Df F{x_{i}}(\chr{S^{t}}) =\sum_{i}\delta_{i}w_{i},
\end{multline*}
where the third equality follows from the fact that $F$ is linear
when one coordinate is changed (Lemma \ref{lem:multilinear}), while
the inequality follows from the fact that the partial derivatives
of $F$ are coordinate-wise non-increasing (Lemma \ref{lem:partialF})
and $y^{i}-\delta_{i}\chr i\ge\chr{S^{t}}$ for all $i$. Thus, we
have: 
\[
\ex{F(y)-F(y-\delta)}\le\ex{\sum_{i}\delta_{i}w_{i}}=\sum_{i}\ex{\delta_{i}}\cdot w_{i}.
\le \frac{1}{\Sigma}(\kout + \kin)\sum_{i}p_{i}x^{t}_{i}w_{i},
\]
where the last inequality follows, as in the linear case, from inequality \eqref{eq:linearloss}.

Returning to the first expectation in \eqref{eq:big-nasty-expectation}, we note that:
\[
\ex{F(\chr{S^{t+1}})-F(\chr{S^{t}})}=\sum_{i}\pr{i\mbox{ probed}}p_{i}(F(S^{t}+i)-F(S^{t}))=\frac{1}{\Sigma}\sum_{i}p_{i}x_{i}^{t}w_{i}.
\]
Thus, the total expected decrease $\ex{z^{t}-z^{t+1}}$ from one
step of our rounding procedure is at most: 
\[
\frac{1}{\Sigma}\sum_{i}p_{i}x_{i}^{t}w_{i}+\frac{1}{\Sigma}\sum_{i}(\kout+\kin)p_{i}x_{i}^{t}w_{i}=(\kout+\kin + 1)\frac{1}{\Sigma}\sum_{i}p_{i}x_{i}^{t}w_{i}.
\]
 On the other hand, the expected increase of $f(S^{t+1})-f\br{S^{t}}$
in this step is: 
\[
\frac{1}{\Sigma}\sum_{e\in E}p_{e}x_{e}^{t}(f(S^{t}+e)-f(S^{t}))=\frac{1}{\Sigma}\sum_{e\in E}p_{e}x_{e}^{t}(F(\chr{S^{t}+e})-F(\chr{S^{t}}))=\frac{1}{\Sigma}\sum_{e\in E}p_{e}x_{e}^{t}w_{e}.
\]
 Thus, by Lemma \ref{lem:martingale}, the final solution $S^{\tau}$
produced by the algorithm satisfies 
\[
\ex{f(S^{\tau})}\ge\left(1-\frac{1}{e}\right)\left(\frac{1}{\kout+\kin+1}\right)\ex{OPT}.
\]

\bibliography{submodularprobing}

\end{document}